\documentclass[runningheads]{llncs}
\usepackage{amsmath}
\usepackage{amsfonts}
\usepackage{amssymb}
\usepackage{authblk}

\usepackage{graphicx}
\usepackage{hyperref}
\usepackage[capitalize, noabbrev]{cleveref}
\usepackage[autostyle=true]{csquotes}
\usepackage{color}
\usepackage[textsize=footnotesize]{todonotes}

\newcommand{\R}{\mathbb{R}}
\newcommand{\N}{\mathbb{N}}

\newcommand{\set}[1]{\{ #1 \}}
\newcommand{\fromto}[2]{\set{#1, \ldots, #2}}

\newcommand{\bigO}{\mathcal{O}}
\newcommand{\dotunion}{\mathbin{\dot{\cup}}}

\newcommand{\True}{\textsc{True}}

\DeclareMathOperator{\dist}{\text{dist}}
\DeclareMathOperator{\rel}{\text{rel}}
\DeclareMathOperator{\betw}{\text{betw}}
\DeclareMathOperator{\valid}{\text{valid}}

\newcommand{\I}{\mathcal{I}}
\newcommand{\J}{\mathcal{J}}
\newcommand{\F}{\mathcal{F}}
\DeclareMathOperator{\graph}{G}

\newcommand{\claimqed}{\hfill\scriptsize$\blacksquare$\normalsize}

\begin{document}
\title{Assistance and Interdiction Problems on Interval Graphs}
%
%
\author{Hung P. Hoang\inst{1} \and
Stefan Lendl\inst{2} \and
Lasse Wulf\inst{3}}
\authorrunning{H. P. Hoang, S. Lendl, L. Wulf}
%
\institute{Department of Computer Science, ETH Zurich, Switzerland 
\email{hung.hoang@inf.ethz.ch}\\
\and Department of Operations and Information Systems, University of Graz, Austria 
\email{stefan.lendl@uni-graz.at}\\
\and Institute of Discrete Mathematics, Graz University of Technology, Austria\\
\email{wulf@math.tugraz.at}}
\maketitle              
\begin{abstract}
We introduce a novel framework of graph modifications specific to interval graphs. 
We study interdiction problems with respect to these graph modifications.
Given a list of original intervals, each interval has a replacement interval such that either the replacement contains the original, or the original contains the replacement.
The interdictor is allowed to replace up to $k$ original intervals with their replacements. 
Using this 
framework we also study the contrary of interdiction problems which we call assistance problems. 
We study these problems for the independence number, the clique number, shortest paths, and the scattering number. 
We obtain polynomial time algorithms for most of the studied problems. Via easy reductions, it follows that on interval graphs, 
the most vital nodes problem with respect to shortest path, independence number and Hamiltonicity can be solved in polynomial time.
\end{abstract}
\section{Introduction}


\emph{Network interdiction} \cite{NetworkInterdictProblemsBookChapter} is a vibrant and fast-growing area of research. 
A network interdiction problem is a min–max or max–min problem 
involving a graph parameter $\pi$ and two opposing players: 
The goal of the second player, the \emph{network owner}, is to optimize $\pi$, while the goal of the first player, the \emph{interdictor}, 
is to alter the network in such a way to maximally impair the owner's objective. 
Interdiction problems are fundamental problems in network 
analysis, because they tell us, which parts of a network are most susceptible 
to failure or attack \cite{criticalNodeDetectionSurvey}. 
One important type of network interdiction problems
is the \emph{most vital nodes problem} with respect to different graph parameters $\pi$. 
Here, we are given a graph, and some natural number $k$ called the \emph{budget}. The interdictor is allowed to delete up to $k$ vertices from the graph in order to impair $\pi$. For example, if the parameter $\pi = \alpha$ is the independence number, then the interdictor seeks to delete up to $k$ vertices from the graph, such that the size of the largest independent set in the remaining graph is minimized. The problem is to compute which vertices to select.

The most vital nodes problem has been considered for many different parameters, and many different special graph classes. (See e.g.\ \cite{baier2010length,complexityOfFindingMostVitalNodesShortestPath,mostVitalNodesWrtIndSet,mostVitalLinksNodes1982,mahdavi2014minimum}.) Diner et al. \cite{diner2018contractionDeletionBlockers} initiated the investigation of the most vital nodes problem on interval graphs. 
With respect to the clique number, they proved that a simple greedy-algorithm 
solves the most vital nodes problem on interval graphs. However, with respect to the independence number $\alpha$, they could not determine the complexity of the most vital nodes problem on interval graphs and left it as an open question~\cite[(Q2)]{diner2018contractionDeletionBlockers}.

In this paper, we positively answer (Q2): On interval graphs, the most vital nodes problem with respect to $\alpha$ can be solved in polynomial time.  Moreover, we extend this result to a much more general framework. For many graph parameters, the most vital nodes problem is a special case of the framework.

\subsection{The Shrink-Expand Framework} 
We propose a new framework, the \emph{shrink-expand framework}, for interdiction problems specifically 
for interval graphs. It revolves around \emph{shrinking} and \emph{expanding} intervals. Here, one is given two intervals $I$ and $I'$, such that either $I' \subseteq I$, or $I' \supseteq I$. The interval $I$ is called the \emph{original interval}, and the interval $I'$ is called the \emph{replacement interval} .
\begin{itemize}
\item If $I' \subseteq I$, then the operation of substituting the original interval $I$ with the smaller interval $I'$ is called \emph{shrinking the interval.}
\item  If $I' \supseteq I$, then the operation of substituting the original interval $I$ with the bigger interval $I'$ is called \emph{expanding the interval.}
\end{itemize}
 
We study interdiction problems with respect to shrinking and expanding intervals. As an input, we are given a fixed list $I_1,\dots,I_n$ of original intervals, and a fixed list $I'_1,\dots,I'_n$ of replacement intervals. For each $j=1,\dots,n$, the interval $I'_j$ is the replacement of the interval $I_j$, and we have $I'_j \subseteq I_j$ or $I'_j \supseteq I_j$.
Furthermore, we are given a budget $k \geq 0$. We start with the interval graph on the original intervals $I_1,\dots,I_n$.
In order to impair $\pi$, the interdictor can then choose a set of at most $k$ intervals and shrink (expand) the chosen intervals.
The question for the interdictor is now, which intervals to select in order to maximally impair $\pi$. For all the graph parameters considered in this paper, only one of shrinking or expanding makes sense. For example, in the clique number interdiction problem, the network owner seeks to find a clique of maximal size, and the interdictor impairs this objective by shrinking intervals. In the next section, we will list for each problem considered, whether one has shrinking or expanding intervals.

For many graph parameters, the shrink-expand framework contains the most vital nodes problem as a special case for the following reason: Consider a pair of original interval $I_j$ and its replacement interval $I'_j$, such that the replacement interval $I'_j = \emptyset$ is empty. The act of shrinking $I_j$ down to $I'_j$ corresponds to removing all edges from the corresponding vertex $v$ in the interval graph. For the graph parameters of  clique number and shortest path, it is easily seen that removing all incident edges from a vertex is equivalent to deleting it from the graph. Therefore, if every single replacement interval is empty, we have exactly the most vital nodes problem. For the graph parameter of independent set, an analogous observation holds. (Here, one chooses the replacement interval so large, that it intersects every other interval.)

While interdiction problems are min-max or max-min type of problems, the shrink-expand framework allows us to also consider \emph{assistance problems}, which are min-min or max-max type of problems. An assistance problem is a two-player problem, where the first player selects at most $k$ intervals to shrink (expand), and the second player optimizes some graph parameter $\pi$. But this time, the two players have a shared objective, instead of a conflicting one. For example, in the assistance problem for the clique number, one has expanding intervals, and the question is which intervals to expand in order to get an interval graph with maximum possible clique number.
If the interdiction problem has shrinking intervals, the assistance problem has expanding intervals, and vice versa. Assistance problems are interesting to consider, as they are the natural counterpart to interdiction problems. 

In this paper, we consider the interdiction and assistance problem for the following 
four classical graph parameters: 
Independence number $\alpha$, maximum clique size $\omega$, shortest path, 
and the scattering number (a graph parameter determining the Hamiltonicity property of interval graphs). 

\subsection{Related Work}
The area of interdiction problems on graphs has gained significant attention in 
the computer science literature (see eg.~\cite{NetworkInterdictProblemsBookChapter}).
The \emph{most vital nodes} problem for graph parameter $\pi$ is the problem where the interdictor is given a budget of $k$ vertex deletions and wants to maximally impair $\pi$. (In the earlier literature, the name only refers to the case $\pi =$  shortest path \cite{mostVitalLinksNodes1982}.)  In the very closely related \emph{vertex blocker} problem, the interdictor is given a threshold $t < \pi(G)$, and we wish to compute the required budget $k$ to reduce $\pi(G)$ down to $t$ \cite{shortPathsInterdictionTotalAndNodeWise}. In this paper, the difference is of no concern: We consider $k, t$ as part of the input, and in this case it is clear that if one problem is in $P$, so is the other. Many different parameters $\pi$ have been considered for both problems, see e.g.\ Nasirian et al. for a short summary \cite{nasirian2019exact}.
It is easy to see that the most vital nodes problem can be modeled as a special case of the interdiction problem in the shrink-expand framework.
Lewis and Yannakakis proved that for every hereditary and nontrivial graph parameter $\pi$, 
the most vital node problem is NP-complete in general graphs \cite{metaTheoremHereditary}.
 
The most vital nodes problem for parameters $\alpha$ and $\omega$ are well studied, 
even for special graph classes~\cite{mostVitalNodesWrtIndSet,costa2011minimum},
but interval graphs have not yet been investigated. Also, the most vital nodes 
problem for shortest paths has a long history~\cite{baier2010length,complexityOfFindingMostVitalNodesShortestPath,mostVitalLinksNodes1982,shortPathsInterdictionTotalAndNodeWise}.
When considering the most vital \emph{edges}, instead of the most vital nodes, there are results for interval graphs with respect to the parameter of shortest path. Bazgan et al.\ conjectured this problem to be polynomially solvable for unit interval graphs \cite{bazgan2019more} and Bentert et al.\ proved this conjecture \cite{bentert2019lengthboundedcuts}.



\subsection{Our Contribution} 
We introduce a novel framework of graph modifications  
specific to interval graphs. The most vital nodes problem for interval graphs can be reduced to  this framework for all the graph parameters studied in this work. We observe that under this model not only interdiction but also the 
notion of an assistance problem can be defined. 
We study these two problem types for the four graph parameters $\alpha$, $\omega$, shortest path, and scattering number. The scattering number of an interval graph is related to its hamiltonian properties. In particular, we look at the \emph{interdiction problem with respect to Hamiltonicity}. Here, we are given an interval graph, which is hamiltonian, and the interdictor wishes to modify the graph such that it is not hamiltonian anymore.
We also obtain similar results for the related parameters of the path cover number, and graph toughness.

Even though the framework is general, we obtain polynomial-time algorithms
for 6 of the 8 studied problems. We obtain these results using (in some cases technically challenging) dynamic programs. 

\begin{theorem}
The following problems can be solved in polynomial time:
\begin{itemize}
\item The interdiction problem in the shrink-expand framework, for the parameters $\alpha$,  shortest path and Hamiltonicity.
\item The assistance problem in the shrink-expand framework, for the parameters $\omega$, $\alpha$, shortest path, and scattering number.
\item The most vital nodes problem on interval graphs, for the parameters $\alpha$, shortest path, and Hamiltonicity.
\end{itemize}
\end{theorem}

The interdiction problem in the shrink-expand framework for the parameter $\omega$ is shown to be NP-complete and $W[1]$-hard. \cref{tab:result-overview} provides an overview of our results. 
Our results for the independence number answer the question (Q2) stated by Diner et al. \cite[(Q2)]{diner2018contractionDeletionBlockers}. Additionally, in Appendix~\ref{app:contraction}, we show how to modify our proof to also answer their question (Q1), regarding edge contraction blockers on interval graphs~\cite[(Q1)]{diner2018contractionDeletionBlockers}.

\begin{table}
\centering
\setlength{\tabcolsep}{6pt}
\begin{tabular}{r|rrr}
$\pi$ & assistance & interdiction & most vital nodes\\
\hline
$\omega(G)$ & $\bigO(n)$ & \ NP-complete, W[1]-hard & $\bigO(n)$ \cite{diner2018contractionDeletionBlockers} \\
$\alpha(G)$ & $\bigO(kn)$ & $\bigO(kn^2)$ & $\bigO(k n^2)$\\
Hamiltonicity & ? & $\bigO(kn^3)$ & $\bigO(k n^3)$\\
shortest path & $O(kn)$ & $O(n^4)$ & $\bigO(n^4)$
\end{tabular}
\caption{Overview of the results}
\label{tab:result-overview}
\end{table}

\section{Preliminaries}
Let $[n] = \fromto{1}{n}$ denote the set of the first $n$ positive natural numbers and $\mathbb{N}_0 = \mathbb{N} \cup \{ 0 \}$ denote the set of nonnegative integers. We use standard notation and terminology from graph theory to denote 
the common graph parameters studied in this work. 
Given a graph $G = (V,E)$ a  subset $S$ of $V$ is called a \emph{clique} if 
for each pair of distinct vertices $u,v \in S$, it holds that $\{u,v\} \in E$.
Analogously, a subset $S$ of $V$ is called an \emph{independent set} if for each pair 
of distinct vertices $u,v \in S$, it holds that $\{u,v\} \notin E$.  We denote by $\omega(G)$ the size of a maximum 
clique in $G$ and by $\alpha(G)$ the size of a maximum independent set in $G$.
A sequence $P = (v_0, v_1, \dots, v_{\ell})$ of pairwise distinct vertices of $G$ is called a \emph{path} in $G$ 
if each consecutive pair of vertices is connected by an edge.
We say $P$ starts in $v_0$ and ends in $v_{\ell}$. The number $\ell$ of 
such edges is called the length of the path. Given two vertices $s$ and $t$, 
the shortest path from $s$ to $t$ is a path starting in $s$, ending in $t$, and
of minimum length.
A sequence $C = (v_0, v_1, \dots, v_{\ell})$ of pairwise distinct vertices of $G$
is called a \emph{cycle} if, in addition to each consecutive pair of vertices being connected by an edge, 
also $\{v_{\ell}, v_0\} \in E$. A \emph{Hamilton path} in $G$ is a path consisting of all vertices of $G$
and a \emph{Hamilton cycle} in $G$ is a cycle consisting of all vertices of $G$. We 
call a graph $G$ \emph{hamiltonian} if there exists a Hamilton cycle in $G$.
A \emph{path cover} of $G$ is a set of vertex-disjoint paths $P_1, \dots, P_k$ 
such that each vertex of $G$ is contained in at least one path. The number $k$ of 
such paths is called the size of the path cover. Note that a path cover of size $1$ 
is a Hamilton path. If $S \subseteq V$ is a set of vertices, we denote by $G-S$ the graph obtained from $G$
by deleting the vertices in $S$ (and their incident edges) from $G$.

For $a < b$, the \emph{closed interval} $[a, b]$ is defined as $\set{x \in \R \mid a \leq x \leq b}$. It has \emph{startpoint} $a$ and \emph{endpoint} $b$. We furthermore define $[a, a] := \emptyset$. Note that this is non-standard, but it will turn out convenient to use. We only consider closed intervals in this paper. If $\mathcal{J} =  (J_1, \dots, J_n)$ is a sequence of intervals (possibly containing duplicates), then we denote by $\graph(\mathcal{J}) = (V,E)$ the corresponding \emph{interval graph}
where $V = [n]$ and for any $u,v \in V$, it holds that $\{u,v\} \in E$ if and only if 
$J_u \cap J_v \neq \emptyset$. Note that in the whole work, we identify the set of 
vertices of the graph with the integers from $1$ to $n$. Given a subset of vertices $X \subseteq [n]$, we denote by 
$\J[X] = \{ J_i \colon i \in X \}$ the set of corresponding intervals.
The sequence $\mathcal{J}$ is called an \emph{interval representation} of the graph $\graph(\mathcal{J})$. It is folklore that if $\mathcal{J}$ contains two intervals which have a common start- or endpoint, one can easily obtain a representation $\mathcal{J}'$ such that no two intervals have a common start- or endpoint and $\graph(\mathcal{J}') = \graph(\mathcal{J})$. We will therefore assume this property whenever it is convenient.

\begin{table}
\centering
\begin{tabular}{r|rr}
$\pi$ & assistance & interdiction\\
\hline
$\omega(G)$ & expand & shrink\\
$\alpha(G)$ & shrink & expand \\
Hamiltonicity & expand & shrink \\
shortest path & expand & shrink \\
scattering number & shrink & expand
\end{tabular}
\caption{Problems in the shrink-expand framework have either shrinking or expanding intervals.}
\label{tab:shrink-expand-overview}
\end{table}
A problem in the shrink-expand framework is defined by an input $(\I, \I', n, k)$. Here, $n$ is the number of vertices, $k \geq 0$ is the budget, $\I = (I_1, \dots, I_n)$ is the sequence of original intervals, and $\I' = (I'_1, \dots, I'_n)$ is the sequence of replacement intervals. The interval $I'_j$ is the replacement interval of the original interval $I_j$ for $j \in [n]$.  We have either $I'_j \subseteq I_j$ for all $j$ (\emph{shrinking} case), or $I'_j \supseteq I_j$ for all $j$ (\emph{expanding} case). \cref{tab:shrink-expand-overview} explains for all problems in the shrink-expand framework, whether we are in the shrinking case or the expanding case. Throughout the paper, we always use the letters $a$ and $b$ for start- and endpoints of $I_j$ and $I'_j$, that is $I_j = [a_j, b_j]$, and $I'_j = [a'_j, b'_j]$.

For a set $X \subseteq [n]$ of indices, we denote by $\I_X$ the sequence of intervals that is obtained from $\I$ after shrinking (expanding) the intervals with indices in $X$. Formally, 
\[(\I_X)_j = \begin{cases} I'_j & \text{if }j \in X \\
I_j & \text{if } j \not\in X. \end{cases}\]
When it is convenient, we can assume that each of $\I, \I', \I_X$ contains no duplicates and treat them as sets instead of sequences. We define $G_X := \graph(\I_X)$. We are interested in problems of the form 
\[ \max_{X \subseteq [n], |X| \leq k} \pi(G_X) \ \text{ or } \min_{X \subseteq [n], |X| \leq k} \pi(G_X)\]
where $\pi$ is some graph parameter. The problem is an \emph{interdiction problem} if it is a problem of max-min or min-max type. The problem is an \emph{assistance problem} if it is a problem of min-min or max-max type.

\section{Shortest Path}

In the shortest path interdiction problem, we are given a graph, and two specified vertices $v_s, v_t$. The interdictor wishes to maximize the length of a shortest path connecting $v_s$ and $v_t$. For interval graphs, we consider the following equivalent problem formulation: 
Given a set $Z$ of intervals and two numbers $s, t \in \R$, define $\dist_Z(s, t)$ to be the minimum number of intervals from $Z$, which one needs to \enquote{walk} from $s$ to $t$. (Here, by a walk, we mean a sequence $(X_1, \dots, X_d)$ of intervals, such that $s \in X_1$, $t \in X_d$, and $X_i$ intersects $X_{i+1}$ for $i \in \fromto{1}{d-1}$. For example, in \cref{fig:shortest_path_ex}, we have $\dist_Z(0, y) = 3$ with the corresponding walk $(I_1,I_3,I_4)$.) 
For the interdiction problem, 
the input consists of the set $\I = \{ [a_1,b_1], \dots, [a_n, b_n] \}$ of original intervals, the set 
$\I' = \{ [a_1,b_1], \dots, [a'_n, b'_n] \}$ of replacement intervals, two numbers $s, t \in \R$, and a budget $k \in \N_0$. We are in the shrinking case and want to determine 
\[\max\set{\dist_{\I_X}(s, t) : X \subseteq [n], |X| \leq k}.\] This problem variant is clearly equivalent to the problem described above. Furthermore, without loss of generality, we can assume that $s = 0 < t$, and that for all intervals $F \in \I \cup \I'$ one has $F \subseteq [0, t]$. (This is because it is never beneficial to walk to the left of $s$ or to right of $t$, so if some $F \in \I \cup \I'$ is not contained in $[s, t]$, we can replace $F$ by $F \cap [s, t]$ and still obtain an equivalent problem instance.) 

In the assistance problem, we are analogously given expanding intervals $\I, \I'$ and want to minimize the distance between $s$ and $t$ by expanding at most $k$ intervals. We show in this section that both problems can be solved in polynomial time.

\begin{figure}
\centering
\includegraphics[scale=1]{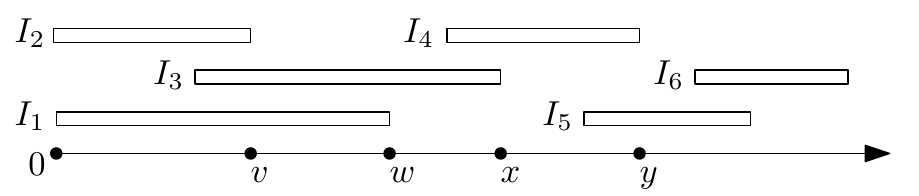}
\caption{Some concepts used in the algorithm for shortest path interdiction.}
\label{fig:shortest_path_ex}
\end{figure}

\subsection{Assistance Problem}
As explained above, we are given expanding intervals $\I, \I'$ and want to determine 
\[ \min \{ \dist_{\I_X}(s, t) \colon X \subseteq [n], |X| \leq k \}. \]
There is a nice analogy for this problem: Suppose we are located on the real number line at point $s$ and want to get to point $t$. Furthermore, an interval $[a_i, b_i] \in \I$ is a road from $a_i$ to $b_i$ which we can use for this cause without paying a fee. Likewise, an interval $[a_i', b_i'] \in \I'$ is a road from $a'_i$ to $b_i'$, but it costs us a fee of one dollar to use. The assistance problem asks for the lowest amount of roads needed to travel from $s$ to $t$ while paying at most $k$ dollars. This can be solved using dynamic programming. 
\begin{theorem}

The assistance problem of shortest path in the shrink-expand framework can be solved in polynomial time. The algorithm can be implemented in $\bigO(kn)$ time, if we are given the list of all intervals sorted both with respect to startpoints and endpoints.
\end{theorem}

\begin{proof}
Consider the set $\mathcal{P} := \bigcup_{i=1}^n \set{a_i, b_i, a'_i, b'_i}$ consisting of all the start- and endpoints of intervals $F \in \I \cup \I'$. For each point $p \in \mathcal{P}$, let $\alpha(p)$ be the rightmost point one can reach from $p$ by using a single interval from $\I$. Formally, $\alpha(p) := \max\set{x \in \mathcal{P} \mid \exists I \in \I: p \in I \text{ and } x \in I}$ or $\alpha(p) = -\infty$ if the set is empty. Likewise, let $\alpha'(p)$ be the rightmost point one can reach from $p$ by using a single interval from $\I'$, i.e.\ $\alpha'(p) := \max\set{x \in \mathcal{P} \mid \exists I' \in \I': p \in I' \text{ and } x \in I'}$ or $\alpha'(p) = -\infty$ if the set is empty. We claim that the set of all values $\alpha(p), \alpha'(p)$ for $p \in \mathcal{P}$ can be precomputed in linear time using a sweep line algorithm, provided that we already have the intervals sorted with respect to both the start- and endpoints. 

Indeed, consider a sweep line starting at $-\infty$ going to $\infty$. We keep track of the fact whether some interval $I \in \I$ intersects the line, and if so, which of these intervals expands the most to the right. Likewise, we keep track of the fact whether some interval $I' \in \I'$ intersects the line, and if so, which of these intervals expands the most to the right. Whenever we encounter some $p \in \mathcal{P}$, we set $\alpha(p)$ and $\alpha'(p)$ and possibly update our status. It is easy to see that this algorithm is correct, and can be implemented in $\bigO(n)$ time. 

Finally, consider the digraph $G$ on vertex set $\mathcal{P} \cup \set{-\infty}$ and the following arc set: For every $p \in \mathcal{P}$, it has one arc $e = (p, \alpha(p))$ and one arc $e' = (p, \alpha'(p))$. We call $e'$ an \emph{expensive} arc. Observe that $G$ is acyclic and has $\bigO(n)$ arcs. Furthermore, the solution to the assistance problem is given by the shortest path in $G$ from $s$ to $t$, which uses at most $k$ expensive arcs. But such a path can be computed using dynamic programming in $\bigO(kn)$ time, via the recursion formula
\begin{align*}
f(p, k') = \min\set{\min\set{f(p',k')+1 \mid (p',p) \text{ not expensive}},\\ \min\set{f(p',k'-1)+1 \mid (p',p) \text{ expensive}}} 
\end{align*}
with starting conditions $f(s, k') = 0$ for all $k' \in \fromto{1}{k}$ and $f(p, -1) = \infty$ for all $p \in \mathcal{P}$.
It is easy to show via induction that $f(p,k')$ is the length of the shortest path from $s$ to $p$
using at most $k'$ expensive arcs. At the end $f(t,k)$ is the optimal objective value for the given assistance problem.
\qed\end{proof}

\subsection{Interdiction Problem}
In this section, we prove: 

\begin{theorem}
\label{thm:shortest_path_interdiction}
The interdiction problem for shortest path in the shrink-expand framework can be solved in time $\bigO(n^4)$.
\end{theorem}

In order to enhance the readability of the proof, we employ the following strategy: First, we show that it suffices to solve the most vital nodes problem instead. Then we prove that the most vital nodes problem can be solved in polynomial time. Finally, we give an argument how the runtime can be improved.

\begin{lemma}
\label{lemma:shortest_path_reduction}
The interdiction problem for shortest path in the shrink-expand framework can be reduced to the most vital nodes problem for shortest path.
\end{lemma}
\begin{proof} 
Let $(\I, \I', k, s,t)$ be an instance of the interdiction problem, where $\I = \fromto{[a_1, b_1]}{[a_n, b_n]}$, $\I' = \fromto{[a'_1, b'_1]}{[a'_n, b'_n]}$, and $[a'_j, b'_j] \supseteq [a_j, b_j]$ for all $j \in \fromto{1}{n}$. Furthermore, $k \in \N_0$ is the budget, and $s, t \in \R$. We construct an equivalent instance of the most vital nodes problem. This instance is described by $(M, k, s, t)$, where $k, s, t$ stay the same as before, and $M$ is the following multi-set of intervals: $M$ contains one copy of $\I$ and $k+1$ copies of $\I'$, i.e.\ it has cardinality $n(k+2)$. More formally, for each $i \in \fromto{1}{n}$, let $N(i, k)$ be the multi-set, consisting of $k+1$ copies of $[a_i', b_i']$. The multi-set $M$ is then given by
\[M = \I \cup \bigcup_{i=1}^n N(i, k). \]

We now claim that the interdiction problem for $\I, \I'$ has the same solution as the most vital nodes problem for $M$. In other words, we have the equality
\[\min_{X \subseteq [n], |X| \leq k} \dist_{\I_X}(s, t) = \min_{Y \subseteq M, |Y| \leq k} \dist_{M \setminus Y}(s, t). \]
In fact, to see \enquote{$\geq$}, whenever we have on the left-hand side a set $X \subseteq [n]$ describing the indices of intervals which are shrunken by the interdictor, 
we can also delete the same intervals in $M$ to get an equivalent term on the right-hand side. 
To see \enquote{$\leq$}, assume we have a multi-set $Y \subseteq M$ on the right-hand side. 
Observe that there is no $i \in \fromto{1}{n}$ such that $N(i, k) \subseteq Y$, simply because $|N(i, k)| = k+1$, 
but $|Y| \leq k$. But 
then the deletion of $N(i, k) \cap Y$ has no effect on the distance between $s$ and $t$. 
We can therefore assume that $Y \cap N(i,k) = \emptyset$ for all $i$, which implies $Y \subseteq \I$. 
Hence, the interdictor can also shrink the same intervals in $\I$ to get an equivalent term on the left-hand side. This completes the proof.
\qed\end{proof}

Certainly, the reduction described by the previous proof will increase the size of the instance by a factor of roughly $k$. We show at the end of this section how to overcome this problem.

In order to solve the most vital nodes problem, we introduce some concepts. If $Z$ is a set of intervals, $x \in \R_{\geq0}$, and $d \in \N_0$, we say that $x$ is \emph{$d$-critical} with respect to $Z$, if $\dist_Z(0, x) = d$, but $\dist_Z(0, x') > d$ for all $x' > x$. (We explicitly allow $\dist_Z(0, x') = \infty$). For example, in \cref{fig:shortest_path_ex}, $x$ is 2-critical with respect to $\fromto{I_1}{I_6}$.

The rough idea of the algorithm is to apply dynamic programming with respect to the rightmost two critical points. To make this more precise, we define the following concepts: If $\I$ is a set of intervals forming an instance for the most vital nodes problem in interval graphs, and if $x, y \in \R_{\geq 0}$, $x < y$, then we call $\rel(x) := \set{[a_i, b_i] \in \I : a_i \leq x}$ the \emph{relevant intervals} of $x$ and we call $\betw(x, y) := \set{[a_i, b_i] \in \I : x < a_i \leq y}$ the intervals \emph{between} $x$ and $y$. For example, in \cref{fig:shortest_path_ex}, we have $\rel(x) = \fromto{I_1}{I_4}$ and $\betw(x,y) = \set{I_5}$. Finally, if $Z \subseteq \rel(x)$, we say that $Z$ is in state $(x, y, d)$, if with respect to $Z$, we have that $x$ is $(d-1)$-critical and $y$ is $d$-critical. Observe that if a set $Z$ is in state $(x, y, d)$, it does not contain any interval that goes beyond $y$, even though $\rel(x)$ may contain intervals that go beyond $y$. For example, in \cref{fig:shortest_path_ex}, the set $\set{I_2, I_3}$ is in state $(v, x, 2)$, but not in state $(v, w, 2)$. We now prove the following central lemma for our algorithm.

\begin{lemma}
\label{lemma:shortest_path_bellmann_principle}
Let $x, y \in \R_{\geq 0}, x < y$ and $d \geq 2$. A set $Z \subseteq \rel(x)$ is in state $(x, y, d)$ if and only if there exists $w < x$ such that both 
\begin{itemize}
\item $Z \cap \rel(w)$ is in state $(w, x, d-1)$
\item $Z \cap \betw(w, x)$ contains no interval $[a_i, b_i]$ with $b_i > y$ and at least one interval $[a_i, b_i]$ with $b_i = y$.
\end{itemize} 
\end{lemma}
\begin{proof}
Note that $(Z \cap \rel(w)) \dotunion (Z \cap \betw(w, x))$ is a partition of $Z$ into two disjoint parts. For the first direction of the proof, if there exists $w$ such that these two parts have the two described properties, it is easy to see that $Z$ is indeed in state $(x, y, d)$. For the other direction, let $Z$ be in state $(x, y, d)$ and let $w < x$ be the unique point on the real line which is $(d-2)$-critical with respect to $Z$. Because $w$ is $(d-2)$-critical, and $x$ is $(d-1)$-critical, and $y$ is $d$-critical, one can check that both the required properties hold.
\qed\end{proof}

An example of the statement of the lemma can be seen in \cref{fig:shortest_path_ex}. If $Z = \{I_1, I_3, I_4\} \subseteq \rel(x)$, then we can see that $Z$ is in state $(x,y,3)$
and the point $w$ has the properties as described in the lemma.

\begin{theorem}
\label{thm_MVN_shortest_path}
The most vital nodes problem for shortest path on interval graphs can be solved in time $\bigO(n^4)$.
\end{theorem}
\begin{proof}
Let $(\I, k, s, t)$ be an instance of the most vital nodes problem with budget $k$, such that $\I = \fromto{[a_1, b_1]}{[a_n, b_n]}$ and $s = 0$ and $t \in \R_{> 0}$. Let $\mathcal{P} := \set{a_i \mid i \in \fromto{1}{n}} \cup \set{b_i \mid i \in  \fromto{1}{n}} 
$ be the set of their start- and endpoints.
We can assume $\min\mathcal{P} = 0$ and $\max\mathcal{P} = t$.
We can also assume that the interdictor cannot disconnect 0 and $t$, as this can be checked easily.

 For $x, y \in \mathcal{P}$, $x < y$, and $d \in \fromto{1}{n}$, we define
\[f(x,y,d) := \max\set{|Z| : Z \subseteq \rel(x), Z \text{ is in state } (x,y,d)}. 
\]

\textbf{Claim.} \emph{If a table of all values $f(x,y,d)$ for $x,y \in \mathcal{P}, x < y$ and $d \in \fromto{1}{n}$ is given, then the solution to the most vital nodes problem can be computed in time $\bigO(n^2)$.}

\emph{Proof of the claim.} Let $X \subseteq \I$ be an optimal choice of intervals, which the interdictor deletes in the most vital nodes problem, and let $d' = \dist_{\I \setminus X}(0, t)$ be the resulting distance. We further assume $X$ has the minimum size among all such optimal choices. Let $Z := \I \setminus X$ be the remaining intervals. Then with respect to $Z$, there exists a unique $(d' - 1)$-critical point $x \in \mathcal{P}$. In other words, we have that the set $Z \cap \rel(x)$ is in state $(x, t, d')$. Furthermore, because $x$ is already $(d' - 1)$-critical, the interdictor does not need to delete any intervals in $\betw(x, t)$. Combined with the assumption that $X$ has the smallest cardinality among the optimal choices, this implies $\betw(x, t) \subseteq Z$. In total, we have the three properties 
\begin{itemize}
\item[(i)] $Z = (Z \cap \rel(x)) \dotunion \betw(x, t)$,
\item[(ii)] $(Z \cap \rel(x))$ is in state $(x, t, d')$, and
\item[(iii)] $|Z| \geq n - k$.
\end{itemize}
 On the other hand, if one has a point $x \in \mathcal{P}$ and a set $Z \subseteq \I$ and some number $d'$, such that properties (i) - (iii) hold, then it is not hard to see that the deletion of $X = \I \setminus Z$ yields a solution to the most vital nodes problem of value $d'$. We therefore conclude that the optimal solution to the most vital nodes problem is given by
\[
d_{\text{opt}} =  \max\set{d' \in \N : \exists x \in \mathcal{P} \text{ s.t.\ } f(x,t,d') + |\betw(x, t)| \geq n - k}.
\]
With this formula, $d_{\text{opt}}$ can be computed in $\bigO(n^2)$ time. 
Note that the values $|\betw(x,t)|$ can be precomputed in $\bigO(n^2)$ time.
This completes the proof of the claim.
\claimqed


We now show how to compute $f$ using dynamic programming. By \cref{lemma:shortest_path_bellmann_principle}, in order to compute $f(x,y,d)$, one can guess $w \in \mathcal{P}, w < x$ and solve independently the maximization problems for $Z \cap \rel(w)$ and $Z \cap \betw(w,x)$. Formally, for $w,x,y \in \mathcal{P}$, where $w < x < y$, we define two helper functions: First, $\valid(w,x,y) := \True$ if and only if there exists $[a_i, b_i] \in \betw(w,x)$ with $b_i = y$. Second, $\beta(w,x,y) := |\set{[a_i, b_i] \in \betw(w,x) : b_i \leq y}|$. The values of all helper functions can be precomputed in time $\bigO(n^3)$, using $\bigO(n^2)$ iterations of a sweep line algorithm in $\bigO(n)$ time.
We then have the following recursion formula for $d \geq 2$:
\[f(x,y,d) = \max\set{f(w,x,d-1) + \beta(w,x,y) : w \in \mathcal{P}, w < x, \valid(w,x,y) = \True}.
\]
The formula is correct, because for a set $Z$ in state $(x, y, d)$, and a point $w$ as described above, $f(w, x, d-1)$ and $\beta(w, x, y)$ are respectively the maximum size of the two disjoint sets $Z \cap \rel(w)$ and $Z \cap \betw(w, x)$.
Hence the correctness follows from \cref{lemma:shortest_path_bellmann_principle}.
For $d = 1$, the values $f(0, y, 1)$ can be easily computed. (Note that $f(x,y,1) = -\infty$ for $x \neq 0$.) In total, for each of the $\bigO(n^3)$ many choices of $(x, y, d)$, we need $\bigO(n)$ computation time, so the total required time is $\bigO(n^4)$. This completes the proof.
\qed\end{proof}

\cref{lemma:shortest_path_reduction,thm_MVN_shortest_path} together prove that the interdiction problem for shortest path in the shrink-expand framework can be solved in polynomial time $\bigO((nk)^4)$.  Finally, we note that when solving the interdiction problem in the shrink-expand framework, the additional factor introduced in \cref{lemma:shortest_path_reduction} can be avoided, yielding an improved running time of $\bigO(n^4)$.

\paragraph*{Proof of \cref{thm:shortest_path_interdiction}.}
Just like in \cref{lemma:shortest_path_reduction}, let $(\I, \I', k, s,t)$ be an instance of the interdiction problem, and let $M$ be the multi-set described in the reduction from the lemma. Then $|M| = n(k+2)$. We know that if we solve the most vital node problem for $M$, i.e.\ if we apply the dynamic program from \cref{thm_MVN_shortest_path}, then we solve the interdiction problem. In particular, if we consider $M$ as an input to the dynamic program from \cref{thm_MVN_shortest_path}, we have 
\begin{align*}
\rel(x) &= \set{[a_i, b_i] \in M : a_i \leq x}\\
\betw(x, y) &= \set{[a_i, b_i] \in M : x < a_i \leq y}\\
f(x,y,d) &= \max\set{|Z| : Z \subseteq \rel(x), Z \text{ is in state } (x,y,d)}\\
d_{\text{opt}} &=  \max\set{d \in \N : \exists x \in \mathcal{P} \text{ s.t.\ } f(x,t,d) + |\betw(x, t)| \geq |M| - k}\\
\beta(w,x,y) &= |\set{[a_i, b_i] \in \betw(w,x) : b_i \leq y}|.
\end{align*}
Now, we observe that the crucial components required to recursively compute the values $f(x,y,d)$ can be also expressed in terms of $\I, \I'$ instead of in terms of $M$. Using this reformulation, they can be computed faster. In fact, the three terms
\begin{align*}
|\betw(x, y)| &= |\set{[a_i, b_i] \in \I : x < a_i \leq y}| + (k+1)|\set{[a'_i, b'_i] \in \I' : x < a'_i \leq y}|\\
d_{\text{opt}} &=  \max\set{d \in \N : \exists x \in \mathcal{P} \text{ s.t.\ } f(x,t,d) + |\betw(x, t)| \geq n(k+2) - k}\\
\beta(w,x,y) &= |\set{[a_i, b_i] \in \betw(w,x)\cap \I : b_i \leq y}| + \\
&\qquad (k+1)|\set{[a'_i, b'_i] \in \betw(w,x)\cap \I' : b'_i \leq y}|
\end{align*}
are sufficient. The set of all values for these terms can be precomputed in $\bigO(n^3)$ time each, using sweep-line algorithms. The rest of the proof is identical to \cref{thm_MVN_shortest_path}.
\qed

Note that the most vital nodes problem can also be reduced to the interdiction  problem  for  shortest  path in the shrink-expand framework.
Observe that by setting $a'=b'$ for any interval the reduction operation 
on this interval implies that the corresponding vertex becomes isolated. 
In the case of shortest path interdiction isolating a vertex is equivalent to vertex deletion.

\section{Independence Number}

For interval graphs, it is well known that the independence number $\alpha(G)$ is 
equal to the so-called clique cover number $\kappa(G)$ which is defined as follows.
Given a graph $G=(V,E)$ a partition $V_1, \dots, V_\ell$ of the vertex set $V$ is called a 
\emph{clique cover} of $G$ if the sets $V_i$ are pairwise disjoint and the graphs $G[V_i]$ induced by $V_i$
form cliques for each $i=1,\dots,\ell$. We call $\ell$ the size of the clique cover 
and denote by $\kappa(G)$ the minimum size of a clique cover of $G$, the clique cover number of $G$.
On general graphs, computing $\kappa(G)$ is NP-hard.
If $G$ is an interval graph, it is a well-known result that $\kappa(G) = \alpha(G)$ and both 
can be computed in polynomial time~\cite{golumbic2004algorithmic}.

In the following section we will heavily use the following observation.

\begin{proposition}\label{prop:cliquecover}
    Let $\I = (I_1, \dots, I_n)$ be a list of intervals, $G = G(\I)$ an interval graph,
     $C(x) = \{i \in V : x \in I_i\}$ and $(V_1, \dots, V_{r})$ a clique cover of $G$.
    Then for every $i \in [r]$ there exists a point $x_{i} \in \mathbb{R}$
    such that $V_{i} \subseteq C(x_i)$.
     Additionally, we can assume that $x_i$ is an endpoint of an interval in $\I$.
\end{proposition}

\subsection{Interdiction Problem}
\label{sec:independence-interdiction}
In this section, we obtain a polynomial time algorithm based on dynamic programming for the independence number interdiction problem,  $\min_{|X| \leq k} \alpha(G_X)$,
on interval graphs. Note that for this problem, we consider the case of expanding intervals.

Since for interval graphs $\alpha(G) = \kappa(G)$, it holds that 
\[ \min_{X \subseteq [n] \colon |X| \leq k } \alpha(G_X) = \min_{X \subseteq [n] \colon |X| \leq k} \kappa(G_X). \]
Therefore, we can focus on solving the clique cover number assistance problem.
Let $\F := \I \cup \I' = \fromto{F_1}{F_{2n}}$ and $F_j = [c_j, d_j]$. Without loss of generality, we assume that the elements of $\F$ are ordered by endpoints, i.e.\ we have $d_1 < d_2 < \dots < d_{2n}$. 
Due to \cref{prop:cliquecover}, we can assume that in an optimal solution $X \subseteq [n]$, the clique cover $(V_1, \dots, V_r)$ of $G_X$ is in one-to-one correspondance to a  subset of cardinality $r$ of the
positions $d_1, \dots, d_{2n}$.

We define $c(i,j)$ 
as the minimum number of intervals that we need to expand such that there lie no 
intervals strictly between $d_i$ and $d_j$ for all $0 \leq i < j \leq 2n+1$, where $d_0=-\infty$ and $d_{2n+1} = \infty$.
Note that if there exists an $\ell$ for which the expanded interval $I'_\ell$ lies strictly between $d_i$ and $d_j$, 
we set $c(i,j) = \infty$. Formally,
\[ c(i,j) = \begin{cases} \infty & \text{if } \exists t \in [n] \colon d_i < a'_t \leq b'_t < d_j\\ |\{ t \in [n] \colon d_i < a_t \leq b_t < d_j \}| & \text{else}. \end{cases} \]
The main idea for defining $c(i,j)$ is that if in a clique cover of $G_X$ no 
endpoint $d_{\ell}$ for $i<\ell<j$ is selected, then $X$ contains at least $c(i,j)$ original intervals that lie strictly in between $d_i$ and $d_j$.

Define $\F_j := \{F_1, \dots, F_j\}$.
Let $F(j,k')$ be the minimum size of a clique cover of $G(\I_X \cap \F_j)$ subject to 
$|\I[X] \cap \F_j| \leq k'$ and $d_j$ is part of the clique cover. Note that if such a clique cover does not exist, 
we define $F(j,k')=\infty$. 
It is easy to see that
    \[ F(j,k') = \min \{ 1 + F(i, k'-c(i,j)) \colon 0 \leq i < j\text{ and } c(i,j) \leq k' \}. \]
Our dynamic programming algorithm is based on this equality.

Based on this result we can precompute the values $c(i,j)$ in $O(n^2)$ time and 
then run a dynamic program to compute all values $F(j,k')$ for $j=0,\dots,n+1$ and 
$k'=0,\dots,k$. Then it holds that $F(2n+1,k) = \min_{X \subseteq [n] \colon |X| \leq k} \kappa(G_X) + 1$
and we obtain the following result.

\begin{theorem}\label{thm:ind-int}
    The independence number interdiction and the clique cover number assistance problem on interval graphs can be solved in $O(k n^2)$ time.
\end{theorem}

Note that for $a'_i = -M$ and $b'_i = M$ for all $i=1,\dots,n$ and some large constant $M$, our 
problem is equivalent to the most vital nodes problem for independent set which was studied 
under the name independence number blocker problem  
by Diner et al.~\cite{diner2018contractionDeletionBlockers}. They left the complexity 
of this problem as an open question. By Theorem~\ref{thm:ind-int}, there exists 
a polynomial time algorithm for this problem.
\begin{corollary}
    The most vital nodes problem for independence number can be solved in $O(k n^2)$ time.
\end{corollary}

\subsection{Assistance Problem}

The independence number assistance problem, $\max_{|X| \leq k} \alpha(G_X)$,
is equivalent to finding an independent set $\J \subseteq \I \cup \I'$ such that 
$ |\J \cap \I'| \leq k$ of maximum size. Note that for this problem we consider the 
case of shrinking' intervals. We can solve this problem in 
polynomial time by giving a reduction to the longest path problem in a directed acyclic graph (DAG) $D$
(see \cref{figure:indendent_set_assistance} for an illustration).

\begin{figure}[thpb]
\centering
\includegraphics[width=\linewidth]{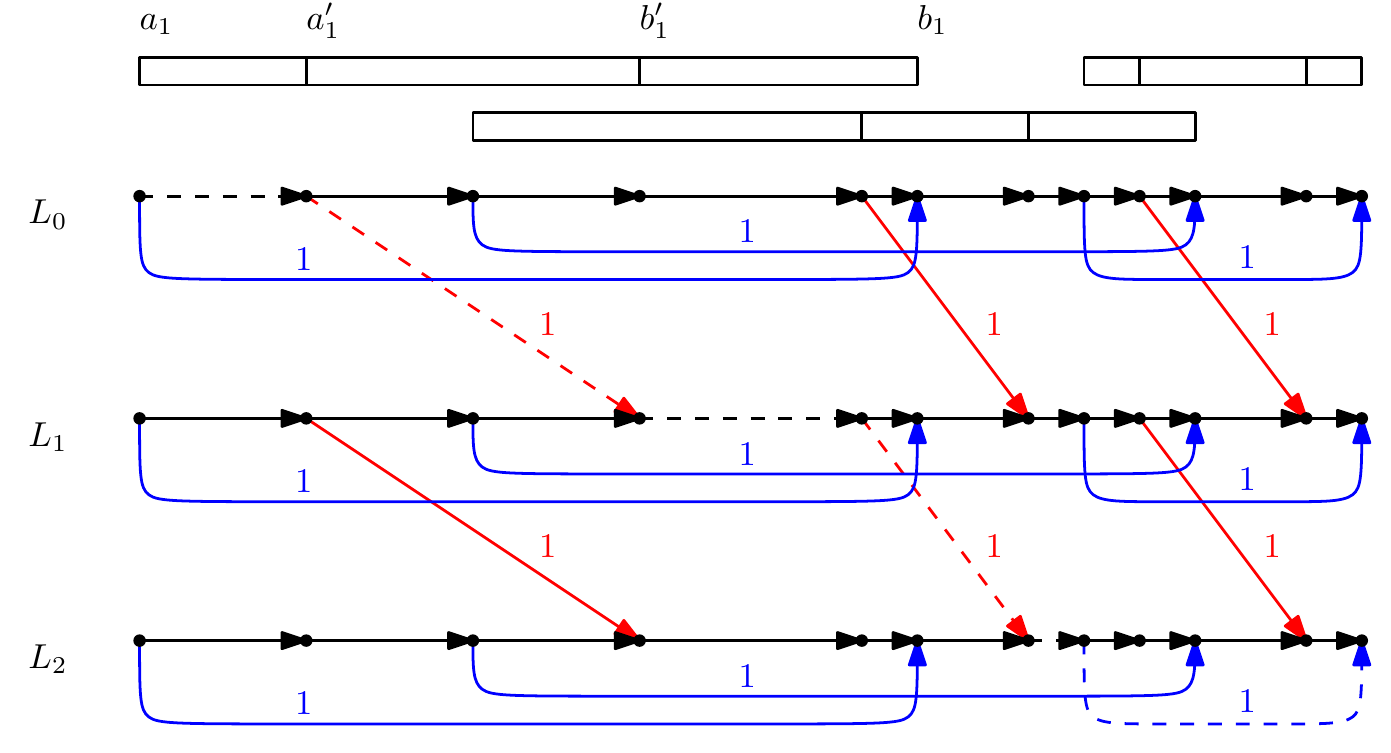}
\caption{Illustration of the digraph used to solve the assistance problem for independent set. We have $n=3$ and $k=2$. 
The blue arcs correspond to the decision to add an interval from $\I$ to the solution and 
the red arcs correspond to the decision to add an interval from $\I'$ to the solution. To increase 
readability, all arcs of the form $(v_i^{\ell}, v_i^{\ell+1})$ are not drawn in the illustration.
The doted edges correspond to an optimal solution for the given instance, where 
intervals $1$ and $2$ are shrunken and interval $3$ is selected without shrinking.}
\label{figure:indendent_set_assistance}
\end{figure}

The vertex set of $D$ consists of $k+1$ layers $L_0, \dots, L_k$. For each start- and endpoint $x \in \mathbb{R}$ of the intervals in $\I$ and $\I'$, we add $k+1$ vertices to $D$,
one in each layer $L_0$ up to $L_k$. For each such vertex $v$, we denote by $p(v) = x$ its corresponding point.
Let $v_1^{\ell}, \dots, v_{2n}^{\ell}$ be the vertices of layer $\ell$ ordered increasingly with respect to 
their corresponding point $p(v_i^{\ell})$ for all $i=1,\dots,2n$ and $\ell=0,\dots,k$.
We add arcs $(v_i^{\ell}, v_i^{\ell+1})$ of length zero to $D$ for all $\ell=0,\dots,k-1$.
We also add arcs $(v_i^{\ell}, v_{i+1}^{\ell})$ of length zero to $D$ for all $i=1,\dots,2n-1$.
For each $i=1,\dots,n$ if $p(v_{i'}^{\ell}) = a_i$ and $p(v_{j'}^{\ell})=b_i$ then we add arcs 
$(v_{i'}^{\ell}, v_{j'}^{\ell})$ of length $1$ for all $\ell=0,\dots, k$ to $D$ (blue arcs in \cref{figure:indendent_set_assistance}). Note that traversing these arcs corresponds 
to adding $I_i$ to the solution. Also, for each $i=1,\dots,n$, if $p(v_{i'}^{\ell}) = a'_i$ and $p(v_{j'}^{\ell})=b'_i$ then we add arcs 
$(v_{i'}^{\ell}, v_{j'}^{\ell+1})$ of length $1$ for all $\ell=0,\dots,k-1$ to $D$ (red arcs in \cref{figure:indendent_set_assistance}). Note that traversing these arcs corresponds 
to adding $I'_i$ to the solution. In this DAG, we are now looking for a longest path from 
$s=v_1^1$ to $t=v_{2n}^k$. Given an $s$-$t$-path $P$ in $D$ we denote by $\I(P)$ 
the set of intervals corresponding to length-1 arcs selected.
Note that for each interval $F \in \I \cup \I'$ 
only one of the length-1 arcs corresponding to adding $F$ to the solution can be 
contained in an $s$-$t$-path, hence $\I(P)$ is well defined.

\begin{lemma}
    There exists an independent set $\J \subseteq \I \cup \I'$ such that 
    $ |\J \cap \I'| \leq k$ if and only if there exists an $s$-$t$-path $P$ 
    in $D$ such that $\I(P) = \J$ of length $|\J|$.
\end{lemma}
\begin{proof}
    All arcs of $D$ go from a vertex corresponding to a point $p$ 
    to a vertex corresponding to a point $p'$ where $p < p'$. In addition,
    all arcs of length $1$ go from the startpoint to the endpoint of the 
    interval that is added to $\I(P)$ if $P$ traverses the arc. Hence, given an
    $s$-$t$-path $P$ the length $1$ arcs traversed by $P$ correspond to 
    pairwise disjoint intervals from left to right, so the intervals $\I(P)$ form an 
    independent set. Also, observe that the number of intervals in $\I(P)$ that 
    are contained in $\I'$ is bounded by $k$, since every $s$-$t$-path can only 
    move from level $\ell$ to level $\ell+1$ once for all $\ell=0,\dots,k-1$.

    It is now easy to see that every independent set $\J \subseteq \I \cup \I'$ such that 
$ |\J \cap \I'| \leq k$ can also be converted into an $s$-$t$-path $P$ such that $\I(P) = \J$.
\qed\end{proof}

Based on this lemma, it holds that computing a longest path in $D$ is 
equivalent to solving the independence number assistance problem. By similar 
arguments as in the subsection above, the same algorithm can also be used to 
solve the clique cover number interdiction problem.
In summary, we obtain the following result.

\begin{theorem}
    The independence number interdiction and clique cover assistance problem can 
    be solved in polynomial time.
    The algorithm can be implemented in $O(kn)$ time if we are given the lists of 
    intervals both sorted with respect to startpoints and endpoints.
\end{theorem}

\section{Maximum Clique Size}
In this section we consider the parameter $\omega$,
the maximum clique size of a graph. Note that for an interval graph $G$ it holds
that $\omega(G) = \chi(G)$, where $\chi(G)$ is the chromatic number of $G$.
The interdiction problem for $\omega$, which is equivalent to the assistance problem 
for $\chi$, has shrinking intervals, 
and the goal is to compute $\min_{|X| \leq k} \omega(G_X)$. 
The assistance problem for $\omega$, which is equivalent to the interdiction problem for $\chi$,
has expanding intervals, and the goal is to compute $\max_{|X| \leq k} \omega(G_X)$.

\subsection{Assistance problem}

\begin{theorem}\label{thm:clique:assistance}
    The maximum clique size assistance problem and the chromatic number interdiction problem can be solved in $\bigO(n)$ time, given the list of all intervals sorted with respect to endpoints. 
\end{theorem}
\begin{proof}
The problem is a highly local problem, since the maximum clique size is achieved at a point 
$x \in \mathbb{R}$ that is included in a maximum number of intervals.
Let $\fromto{p_1}{p_{2n}}$ be the set of all start- and endpoints of intervals in $\I \cup \I'$. For each $i \in [2n]$, let $\alpha_i$ be the number of intervals in $\I$ that contain $p_i$. Let $\beta_i := |\set{j \in [n] : p_i \not\in I_j, p_i \in I'_j}|$. The value $\beta_i$ is the maximum possible value by which $\alpha_i$ could increase by expanding some intervals. The set of all values $\alpha_i, \beta_i$ can be computed using a single sweep line algorithm. The solution to the assistance problem is then given by $\max_{i=1,\dots,2n} (\alpha_i + \min\set{\beta_i,k})$.
\qed\end{proof}

%

\subsection{Interdiction problem}

Diner et al.\ \cite{diner2018contractionDeletionBlockers} proved that the most vital nodes problem for maximum clique size on interval graphs can be solved in $\bigO(n)$ time with a simple greedy algorithm \cite{diner2018contractionDeletionBlockers}. 
Note that the most vital nodes problem for maximum clique size is a special case of the interdiction problem for maximum clique size in the shrink-expand framework 
by setting $a'_i = b'_i$ for all $i$. Observe that in this case the shrinking operation corresponds 
to the operation of removing all incident edges of a vertex, i.e. transforming it into a singleton. For 
the maximum clique size parameter this is equivalent to vertex deletion, except for the special case where $k=n$, which can be trivially handled.
In the general case (where we allow $a'_i \neq b'_i$), this problem becomes much harder. To show this, we reduce from the problem $\textsc{1-fold 2-interval cover}$. In this problem, we are given an integer $t \in \N$ and $n$ objects $O_1,\dots,O_n$, where every object $O_i = [x_{i1},x_{i2}] \cup [x_{i3}, x_{i4}]$ is a disjoint union of two intervals (a so-called \emph{2-interval}), such that $0 \leq x_{i1} < x_{i2} < x_{i3} < x_{i4} \leq t$. The goal is to decide whether there exists a selection of at most $k$ objects such that every $p \in \fromto{0}{t}$ is contained in at least one object of the selection. $\textsc{1-fold 2-interval cover}$ is NP-complete \cite{ding2011onefold} and W[1]-hard \cite{approximability-c-interval}.

\begin{theorem}
The maximum clique size interdiction problem and the chromatic number assistance problem in the shrink-expand framework is NP-complete. The problems are even W[1]-hard with respect to the parameter $k$.
\end{theorem}
\begin{proof}
Given an instance $(t, O_1, \dots, O_n)$ of $\textsc{1-fold 2-interval cover}$, we construct an instance of the maximum clique size interdiction problem. For each object $O_i$, we introduce a shrinking interval pair $[a_i, b_i] = [x_{i1}, x_{i4}]$ and $[a'_i, b'_i] = [x_{i2},x_{i3}]$. Let $\I := \fromto{[a_1, b_1]}{[a_n, b_n]}$ be the set of all $[a_i, b_i]$ introduced so far and let $M := \omega(\graph(\I))$ be the clique number of the corresponding interval graph. We now add for each $p \in \fromto{0}{t}$ additional \enquote{non-shrinkable} pairs $[a_i, b_i] = [a'_i, b'_i] = [p- \epsilon, p + \epsilon]$ for a small $\epsilon > 0$, until $p$ is contained in exactly $M + 1$ intervals. For this newly constructed instance, we have that the interdictor can reduce the clique number from $M + 1$ down to $M$ with a budget of $k$, if and only if there is a 1-fold 2-interval cover of $\fromto{0}{t}$ with at most $k$ objects.
\qed\end{proof}


\newcommand{\scat}{\text{sc}}

\section{Scattering Number, Hamiltonicity, and Graph Toughness}
\label{sec:hamiltonicity}
In this section, we investigate a few parameters related to Hamiltonicity and connectivity of a graph.
We define by $c(G)$ the number of connected components in a graph $G$.
The key parameter for this section is the scattering number of a graph $G$, denoted by $\scat(G)$, that was defined by Jung~\cite{jung1978scat} as follows:

\[
    \scat(G) := \max \{c(G-S) - |S| : S \subseteq V(G) \text{ and } c(G-S) > 1\}.
\]
If the set above is empty, the graph $G$ is a complete graph, and in this case, we define $\scat(G) = -\infty$.

For an interval graph, the scattering number characterizes the Hamiltonicity of the graph, as summarized in the following theorem:

\begin{theorem}[Deogun, Kratsch, and Steiner~\cite{deogun1997scat}]
\label{thm:scat}
For an interval graph $G$ and for all constants $p \geq 1$,
\begin{itemize}
	\item $G$ contains a Hamilton path, if and only if $\scat(G) \leq 1$;
	\item $G$ contains a Hamilton cycle, if and only if $\scat(G) \leq 0$;
	\item $G$ contains a path cover of size $p$, if and only if $\scat(G) \leq p$.
\end{itemize}
\end{theorem}

One can easily see that the ``$\Rightarrow$" directions in the above theorem are true for any graph. In this section, we will prove the following theorem: (the formal definitions of the problems will be presented in the later subsections)

\begin{theorem}
\label{thm:ham}
In the shrink-expand framework, the following problems can be solved in time $\bigO(kn^3)$:
\begin{itemize}
	\item[(a)] The assistance problem for scattering number,
	\item[(b)] The interdiction problem for Hamilton path,
	\item[(c)] The interdiction problem for Hamilton cycle,
	\item[(d)] The interdiction problem for path cover,
\end{itemize}
\end{theorem}

As the scattering number is the central parameter in this section, we will first present the method to solve this assistance problem for scattering number in Section~\ref{subsection:scat}.
After that, we will show how by solving this problem, we can solve the interdiction problems mentioned in Theorem~\ref{thm:ham}.

We also discuss briefly two implications of the above theorem. 
Firstly, in Section~\ref{subsection:ham_mvn}, we show how the most vital nodes problem for Hamiltonicity parameters can be reduced to the assistance problem for scattering number and hence can be solved in similar time. 
Secondly, the method to solve the above assistance problem can be easily extended to the interdiction problem of a related parameter, graph toughness. 
In Section~\ref{subsection:tough}, we will define this parameter and discuss the modification needed to apply the method above to solve the latter problem. 

\subsection{Assistance Problem for Scattering Number}
\label{subsection:scat}

The assistance problem for $\scat(G)$ is to compute $\max_{|X| \leq k}\scat(G_X)$ with shrinking intervals. 
In other words, we want to compute 

\begin{equation}
\label{eq:scat_assist}
  \max_{X \subseteq [n], |X| \leq k} \, \max_{S \subseteq [n], c(G_X - S) > 1} c(G_X - S) - |S|.
\end{equation}

We interprete problem \ref{eq:scat_assist} as a cooperative two-player game between an \emph{assistant} and the network owner. The assistant first selects the indices $X$ of the intervals to shrink, obtaining some graph $G_X$. After that, the network owner selects a set $S$ of vertices in $G_X$. (Note that by definition, the vertex set of $G_X$ is $[n]$.) The goal is that the property $c(G_X - S) > 1$ holds, and that under this condition, the value $c(G_X - S) - |S|$ is as large as possible.

W.l.o.g. we impose the following condition: There is no $i \in X$ such that $i \in S$.
In other words, whenever the assistant chooses to shrink some interval, the network owner does not at the same time decide to delete the corresponding vertex.
Suppose there exists such an $i$.
We then remove $i$ from $X$ while keeping $i$ in $S$. 
Effectively, the interval corresponding to the vertex $i$ is now the original interval instead of the replacement interval.
Observe that $G_X - S$ and $|S|$ remain unchanged.
However, we now have an additional budget to shrink another interval (i.e., adding the corresponding vertex into $X$) and potentially increase $c(G_X - S)$.

The above condition implies that $S \cap X = \emptyset$. 
Further, all the corresponding intervals for the vertices in $S$ are original intervals in $\I$ (and not replacement intervals).
Therefore, we can refer to the set of these intervals as $\I[S]$.

Motivated by the above condition, we have the following definition: A pair $(X,S)$ with $X, S \subseteq [n]$ is called a \emph{candidate pair}, if $|X| \leq k$ and $X \cap S = \emptyset$.
Further, we define 
\[
s(X,S) := \begin{cases} 
c(G_X - S) - |S| &\text{if } c(G_X - S) > 1,\\ 
-\infty &\text{otherwise.}
\end{cases}
\]

By the above explanations, we then have that the value of expression (\ref{eq:scat_assist}) is exactly 
\begin{equation}
 \max_{(X, S) \text{ candidate pair}} s(X, S). \label{eq:scat_assist_prime}
\end{equation}

\subsubsection{The idea.}
In order to show that the assistance problem for the scattering number can be solved in polynomial time, we employ the following proof strategy: The assistant and network owner need to pick a candidate pair $(X, S)$ maximizing expression (\ref{eq:scat_assist_prime}). We will describe a non-deterministic process $P$, which describes a formal, non-deterministic way to pick the elements of $X$ and $S$ one-by-one. We will show that:
\begin{itemize}
\item Every sequence of choices one can take in process $P$ corresponds to a candidate pair $(X, S)$, and also
\item for every candidate pair $(X, S)$, one can find a sequence of choices in process $P$, which corresponds either to $(X, S)$, or to another candidate pair $(X', S')$, which is even \enquote{better} than $(X, S)$. (That is, $s(X', S') \geq s(X, S)$.)
\end{itemize}
We furthermore show that an optimal sequence of choices for process $P$ can be found in polynomial time by reduction to finding the longest path in a DAG. This implies that problem (\ref{eq:scat_assist_prime}) can be computed in polynomial time, which was to be shown.

We describe process $P$ in more detail: For every index $t \in [n]$, we have a pair of original interval $I_t$ and its replacement interval $I'_t$ with $I'_t \subseteq I_t$. For every candidate pair $(X, S)$ we have $X \cap S = \emptyset$. Therefore, the assistant and network owner have to decide for every index $t \in [n]$ on exactly one of the following three options: 

\begin{itemize}
	\item[(i)] $i \in X, i \not\in S$: This means the assistant \emph{shrinks} the original interval $I_t$ down to the replacement $I'_t$ and the network owner does not delete the corresponding vertex. We denote this case by $\textsc{shrink}(I_t)$, $\textsc{lock}(I'_t)$, and we say that $I'_t$ gets \emph{locked}. 
	\item[(ii)] $i \not\in X, i \in S$. This means the assistant does not shrink the original $I_t$, but the network owner \emph{deletes} the corresponding vertex. We also say that the network owner \emph{deletes} $I_t$. We denote this case by $\textsc{delete}(I_t)$, $\textsc{discard}(I'_t)$.
	\item[(iii)] $i \not\in X, i \not\in S$. This means the assistant does not shrink the original $I_t$, and the network owner does not delete the corresponding vertex. We denote this case by $\textsc{lock}(I_t)$, $\textsc{discard}(I'_t)$, and we say that $I_t$ gets \emph{locked}. 
\end{itemize}
During the execution of process $P$, we go through all intervals $F \in \I \cup \I'$ ordered descendingly by their respective endpoint. The process can call $\textsc{lock}(F)$, $\textsc{shrink}(F)$, $\textsc{delete}(F)$ or  $\textsc{discard}(F)$ for each interval. 
This corresponds to the decision whether we are in case (i) -- (iii). 
In same cases we are also allowed to make a partial decision, where we do not immediately decide for one of (i) -- (iii), but instead $\textsc{delay}$  the decision until a later point in time (namely, when encountering $I_t$, it is allowed to delay the decision until encountering $I'_t$). Also, depending on the previous choices, not every choice will be available at every step. 
The process $P$ will have the property, that at the end of the process, for each index $t \in [n]$ exactly one of cases (i), (ii), or (iii) holds. This means that at the end of process $P$ we have constructed some candidate pair $(X, S)$. 
Note that for this pair $(X,S)$, the following holds: The interval representation of the graph $G_X - S$ consists of exactly those intervals $F \in \I \cup \I'$, for which $\textsc{lock}(F)$ has been called. 
In other words, locking an interval $F \in \I \cup \I'$ means committing to the final decision, that this interval will be part of (the interval representation of) $G_X - S$.

Finally, process $P$ will keep some variables in mind, which are initialized at the start, and updated depending on the choices taken. One of these variables, $\beta$, will be such that the final value of $\beta$ at the end of the process equals exactly $s(X, S)$. It then remains to show that the highest possible value of $\beta$ at the end of process $P$ can be computed in polynomial time.

\subsubsection{Preliminaries.}
We describe the concepts necessary for the description of $P$. Let $\mathcal{F} = \I \cup \I' = \fromto{[c_1, d_1]}{[c_{2n}, d_{2n}]}$ be ordered by endpoints, i.e.\ $d_1 < d_2 < \dots < d_{2n}$. 
Process $P$ sequentially considers the interval $F_i = [c_i, d_{i}]$, for $i = 2n, \dots, 1$. 
Note that either $F_i \in \I$ or $F_i \in \I'$.
Further, we can assume that there are no empty intervals in the set $\I$ of original intervals, because if there are any, we can easily preprocess to account for and remove them. 
(However, a small technicality is that one or more of the replacement intervals could be empty. We have to keep in mind that for every  $t \in [n]$ such that $I'_t = \emptyset$ and $t \in X$, the graph $G_X - S$ has a connected component consisting of a single isolated vertex.)

At each point in time $i = 2n, \dots, 1$ during the execution of $P$, let $R_i$ be the set of intervals that have been locked so far. The interval graph $\graph(R_i)$ has multiple components in general. Among all the components of $\graph(R_i)$ that do not correspond to an empty interval, let $R_i' \subseteq R_i$ be the set of intervals corresponding to the leftmost component. The process $P$ keeps in memory two variables $\ell, r$, which describe an interval $L := [\ell, r]$. We will make sure that during the execution ($i=2n,\dots,1$) of $P$, we maintain two invariants: 
\begin{equation}
\label{eq:invariant_lock_in_alg}
L = \bigcup_{F \in R_i'}F \text{ and } r \geq d_i. 
\end{equation}
The first invariant essentially means that $L$ is the leftmost non-empty interval corresponding to a component of $G(R_i)$, while the second invariant naturally follows the order of processing from largest endpoint to the smallest.

In addition, the process keeps in memory a variable $k'$, which starts with value $k$ and describes the remaining budget of the assistant.
The process also keeps in memory a variable $\gamma$ that indicates $\min\{c(G(R_i)),2\}$.
In words, this variable keeps track of how many connected components formed by the locked intervals so far.
Since at the end, we are only interested in whether the number of connected components is at least two, we cap the value of this variable at two.

Finally, the output of the process will be calculated via a variable $\beta$, which starts at value 0.
The idea is that at the end of the process (i.e., after we finish with $i = 1$), $\beta$ will equal exactly $s(X,S)$, where $(X,S)$ is the corresponding candidate pair that the assistant has selected during the process.

\subsubsection{The process $P$.} 

Initialize $L \leftarrow \emptyset$, $k' \leftarrow k$, $\gamma \leftarrow 0$, and $\beta \leftarrow 0$. 

For $i = 2n, \dots, 1$, consider $F_i$. If $F_i \in \I$, i.e.\ if $F_i = I_t$ for some $t$, choose exactly one of the following options:
\begin{itemize}
\item $\textsc{lock}(I_t)$: If $I_t \cap L \neq \emptyset$, then set $L \leftarrow L \cup I_t$. Else, i.e.\ if $I_t \cap L = \emptyset$, then set $\beta \leftarrow \beta + 1$, $\gamma \leftarrow \min\{\gamma + 1, 2\}$, and $L \leftarrow I_t$.
\item $\textsc{delay}(I_t)$: Do nothing.
\end{itemize}
On the other hand, if $F_i \in \I'$, i.e.\ if $F_i = I'_t$ for some $t$, choose exactly one of the following possible options (note that $I'_t \subseteq I_t$):
\begin{itemize}
\item $\textsc{lock}(I_t)$, $\textsc{discard}(I'_t)$: Only possible if $I_t \subseteq L$. Do nothing.
\item $\textsc{delete}(I_t)$, $\textsc{discard}(I'_t)$: Only possible if $I_t \not\subseteq L$. Set $\beta \leftarrow \beta - 1$.
\item $\textsc{shrink}(I_t)$, $\textsc{lock}(I'_t)$: Only possible if $k' > 0$ and $I_t \not\subseteq L$. Set $k' \leftarrow k' - 1$. If $I'_t \cap L \neq \emptyset$, set $L \leftarrow L \cup I'_t$. Else, set $\beta \leftarrow\beta + 1$ and $\gamma \leftarrow \min\{\gamma + 1, 2\}$; set $L \leftarrow I_t'$ if $I'_t \neq \emptyset$.
\end{itemize}

Finally, after all the iterations above, if $\gamma < 2$, set $\beta \leftarrow -\infty$.

\begin{figure}
\centering
\includegraphics[scale=1]{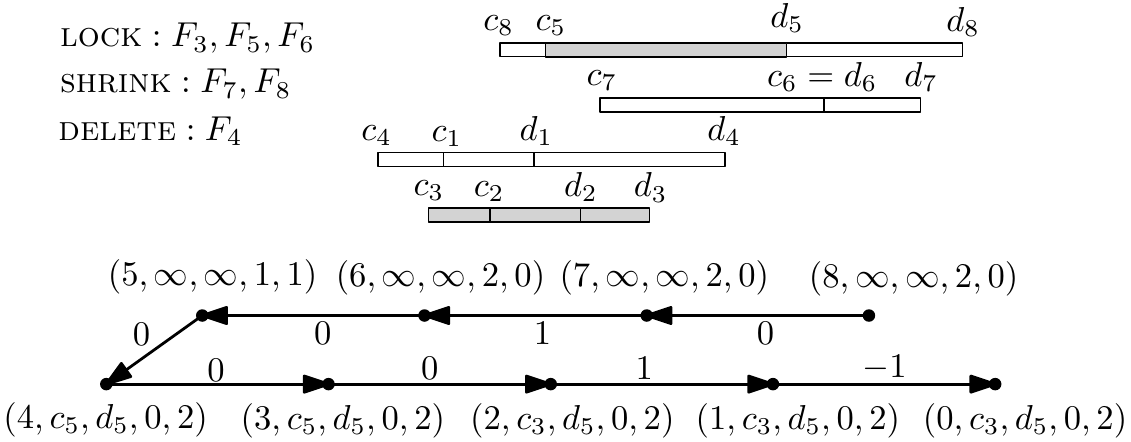}
\caption{An example execution of process $P$: The top half shows the intervals, where each pair of enveloping intervals represents an original interval and its corresponding replacement. The locked, shrunken, and deleted intervals are listed on the left. The bottom half shows the memory states $(i,\ell,r,k,\gamma)$ throughout the process. The numbers on the edges represent the changes in $\beta$.}
\label{fig:ham_path}
\end{figure}

See Figure~\ref{fig:ham_path} for an illustration of an execution of process $P$.
At the beginning, where $L = \emptyset$, we use the convention $\ell = r = \infty$. 
Further, $i = 0$ corresponds to the final memory state of the process.

\begin{lemma}
\label{lem:process_P}
(i) For every sequence of choices one can take in process $P$, there exists a candidate pair $(X,S)$, such that the final value of $\beta$ equals $s(X, S)$.
(ii) For all candidate pairs $(X,S)$, there is a sequence of choices one can take in process $P$, such that at the end of the process, $\beta \geq s(X, S)$. 
\end{lemma} 
\begin{proof}
(i) Firstly, observe that for any execution of $P$, the invariants (\ref{eq:invariant_lock_in_alg}) are maintained throughout the process.
Secondly, we claim that at the end of process $P$, for every $I \in \I$, exactly one of $\textsc{shrink}(I)$, $\textsc{delete}(I)$, or $\textsc{lock}(I)$ was called. 
Indeed, when we process the interval $I$, if $\textsc{delay}(I)$ is called, then when we process the corresponding replacement interval $I'$, exactly one of the above operations is called for $I$.
Now suppose $\textsc{lock}(I)$ is called when we process the interval $I$.
By the design of the process, $I \subseteq L$ after this step.
Due to invariant (1) and the fact that the corresponding replacement interval $I'$ is contained in $I$, right before we process $I'$, we still have $I' \subseteq I \subseteq L$.
Hence, the option $\textsc{lock}(I), \textsc{discard}(I')$ will be called.

Now consider an execution of process $P$ such that $\beta_0$ is the final value of the variable $\beta$.
We define $X$ ($S$) as the set of the indices such that \textsc{shrink} (\textsc{delete} respectively) has been called for the corresponding intervals. 
By the observation above, we have $X \cap S = \emptyset$.
Further, whenever we call the operation $\textsc{shrink}$, $k'$ is decreased by 1.
As $k'$ starts at $k$ and cannot go below 0 during the process, we have $|X| \leq k$.
Hence, $(X, S)$ is a candidate pair.

Consider the set $R_0$ of all intervals $F \in \mathcal{F}$ that were locked during the execution of $P$. 
It is easy to see that each connected component of $\graph(R_0)$ corresponds to exactly one moment in the execution of $P$, where $\beta$ was increased and vice versa. 
Indeed, if the component corresponds to an empty interval $I'_t$, this moment was the iteration $i$, for $F_i = I'_t$.
Otherwise, using invariants (\ref{eq:invariant_lock_in_alg}), the value of $r$ after the step at this moment corresponds to the rightmost point of the component of $\graph(R_0)$. 
Clearly, the amount of times $\beta$ was decreased is exactly $|S|$.
Further, whenever $\beta$ is increased, we also increase $\gamma$ subject to the cap of 2.
Hence, if $\gamma$ is less than 2 after the process, $c(\graph(R_0)) \leq  1$ and $\beta$ is set to $-\infty$. 
Overall, we have $\beta_0 = s(X, S)$. 

(ii) Consider the following execution of the process $P$.
For each $t \in [n]$, when considering $I_t$, we choose the option $\textsc{lock}(I_t)$ if $t \notin X \cup S$ and the option $\textsc{delay}(I_t)$ otherwise.
When considering $I'_t$, we choose $\textsc{lock}(I_t), \textsc{discard}(I'_t)$ if possible.
Otherwise, this implies that $t \in X$ or $t \in S$. 
Because $(X,S)$ is a candidate pair, it is impossible that $t \in X \cap S$.
Then we choose $\textsc{shrink}(I_t), \textsc{lock}(I'_t)$ if $t \in X$ and $\textsc{delete}(I_t),$ $\textsc{discard}(I'_t)$ if $t \in S$.

Following the argument for the proof of (i) above, we can construct $X'$ and $S'$ such that the final value $\beta$ of the execution of $P$ equals $s(X',S')$. 
By construction, $X' \subseteq X$ and $S' \subseteq S$.
Further, following the above execution of the process $P$, it is easy to see that the intervals on the real line corresponding to the connected components of $G_X - S$ are the same as those for $G_{X'} - S'$.
Therefore, $c(G_{X'} - S') - |S'| \geq c(G_{X} - S) - |S|$, and hence $s(X',S') \geq s(X,S)$.
The lemma then follows.
\qed
\end{proof}

\begin{lemma}
\label{lem:ham_path}
The maximum possible value of $\beta$ at the end of process $P$ can be computed in time $\bigO(kn^3)$ by reduction to the maximum-cost path in a DAG.
\end{lemma}
\begin{proof}
An instance of the maximum-cost path in a DAG problem is a tuple $(H, \omega)$, where $H$ is a DAG and a function $\omega: E(H) \to \R$ describes the edge cost of $H$.
The problem is to find a path in $H$ with the maximum cost, which is defined as the sum of edge costs along the path.
This problem can be solved with dynamic programming with runtime linear to the number of edges of H.

We first construct an instance of this problem for the reduction.
For every possible value assignments of $L = [\ell, r]$, $k'$, and $\gamma$ describing the memory state of process $P$ at step $i \in \fromto{2n}{0}$, we add a vertex $(i,\ell,r,k',\gamma)$ in $H$. 
For $i \geq 1$, we create an edge from a vertex of the form $(i,\ell_1,r_1,k'_1,\gamma_1)$ to vertices $(i-1,\ell_2,r_2,k_2',\gamma_2)$, if and only if there is some possible choice in $P$ such that if the current memory state is $(i,\ell_1,r_1,k'_1,\gamma_1)$ then the new memory state would be $(i-1, \ell_2, r_2, k_2',\gamma_2)$ after this choice. 
Observe that the knowledge of $(i, \ell_1, r_1, k'_1,\gamma_1)$ suffices to determine which choices are possible to determine the new memory state. 
The edge costs are given by $+1$, if the line $\beta \leftarrow \beta + 1$ is called, by $-1$, if the line $\beta \leftarrow \beta- 1$ is called.
One exception is the step from a vertex $(1,\ell_1,r_1,k'_1,\gamma_1)$ to a vertex $(0,\ell_2,r_2,k'_2,\gamma_2)$, where $\gamma_2 < 2$.
In this case, we set the corresponding edge cost to $-\infty$.
The costs of all other edges are 0.

By construction, $H$ is a DAG, because for any edge in $H$, the first index of the starting vertex is higher than that of the ending vertex.
Also by construction, a path connecting the starting vertex $(2n, \infty, \infty, k, 0)$ to a vertex $(0, \ell, r, k', \gamma)$ is a maximal path, and its cost is the final value $\beta$ at the end the corresponding execution of $P$ that incurs the memory states along the path.
Consequently, the cost of the maximum-cost path is exactly the maximum possibe value of $\beta$ at the end of process $P$. 
There are $\bigO(n)$ possible values for each of $i, \ell, r$ respectively, so $H$ has $\bigO(kn^3)$ vertices and edges. Hence the desired path can be found in $\bigO(kn^3)$ time.
\qed
\end{proof}

\paragraph*{Proof of Theorem~\ref{thm:ham}(a).}
Let $\beta_{\max}$ be the maximum possible value of $\beta$ at the end of process $P$.
As a consequence of Lemma~\ref{lem:process_P}, there exists a candidate pair $(X,S)$ such that $s(X,S) = \beta_{\max}$, and
for any candidate pair $(X',S')$, there exists an execution of $P$ such that the final value of $\beta$ is at least $s(X',S')$.
This implies that $\beta_{\max}$ is the maximum of $s(X,S)$ over all candidate pairs $(X,S)$, i.e., the value of the original problem that we want to compute.
By Lemma~\ref{lem:ham_path}, $\beta_{\max}$ can be computed in time $\bigO(kn^3)$.
Theorem~\ref{thm:ham}(a) then follows.
\qed

\subsection{Interdiction Problems for Hamiltonicity}

In the interdiction problem for Hamilton path (cycle), we are given shrinking intervals $\I, \I'$, such that $\graph(\I)$ has a Hamilton path (cycle), and the interdictor wishes to select $X \subseteq [n]$ of size at most $k$, such that $\graph(\I_X)$ does not have a Hamilton path (cycle) anymore.
Similarly in the interdiction problem for path cover, we also have shrinking intervals, and the objective of the interdictor is to maximize the size of the path cover within the budget $k$ shrinkages.

\paragraph*{Proof of Theorem~\ref{thm:ham}(b)-(d).}
These are implied by Theorem~\ref{thm:scat} and Theorem~\ref{thm:ham}(a).
To solve these problems, we first compute $\max_{|X| \leq k}\scat(G_X)$, i.e. the value of the assistance problem of the scattering number.
This value is also the value for the interdiction problem for path cover.
For the interdiction problems for Hamilton path and Hamilton cycle, we can conclude that the interdictor can interdict successfully, if and only if this value is more than 1 and 0 respectively.
Therefore, these problems can also be solved in time $\bigO(kn^3)$.
\qed

\subsection{Most Vital Nodes Problem for Hamiltonicity}
\label{subsection:ham_mvn}
Similar to other graph parameters discussed in this paper, a natural question is whether the most vital nodes problems for Hamilton path, Hamilton cycle, and path cover are special cases of the shrink-expand framework.
The simple answer is no.
Here, shrinking to an empty interval and deleting a vertex yield different effect.
The former operation creates an isolated vertex and disconnects the graph, while the latter does not necessarily make the graph non-hamiltonian.

However, the most vital nodes problems for the above parameters can be reduced to the assistance problem for scattering number.
By Theorem~\ref{thm:scat}, the former problems can be immediately solved, if we know the value of the following problem $\textsc{Scat}(G,k)$: Given an interval graph $G$ and a budget $k$, what is the maximum scattering number that can be achieved after $k$ vertex deletion?
Note that this is technically not a most vital nodes problem in the traditional sense, as the goal of two players are aligned instead of conflicting.
The value of the problem $\textsc{Scat}(G,k)$ above in turn can be easily derived from the assitance problem for scattering number, as shown in the following lemma.

\begin{lemma}
Let $G$ be an interval graph and $k$ an integer.
Let $s_1$ be the value of the problem $\textsc{Scat}(G,k)$.
Let $s_2$ be the value of the assistance problem for scattering number with the budget $k$, where the original intervals are the intervals of $G$, and the replacement intervals are all empty intervals.
Then $s_1 = s_2 - k$.
\end{lemma}
\begin{proof}
Consider an optimal solution of $\textsc{Scat}(G,k)$.
There are $k' \leq k$ deleted vertices in this solution.
Replacing these vertices with empty intervals in the assistance problem yields a valid (but not necessarily optimal) solution with value $s_1 + k'$.
Since shrinking an interval always strictly increases the value of the scattering number, we shrink further $k - k'$ other intervals.
The value then becomes $s_1 + k$ and this is still a valid solution.
Hence, $s_1 + k \leq s_2$.

For the other direction, with the same observation that shrinking intervals always make the value better, there must exist an optimal solution of the assistance problem with exactly $k$ shrunk intervals.
We then delete these intervals in the problem $\textsc{Scat}(G,k)$, yielding a valid (but not necessarily optimal) solution with value $s_2 - k$.
Hence, $s_1 \geq s_2 - k$.

The lemma then follows. 
\qed
\end{proof}

Combining the above lemma with Theorem~\ref{thm:ham}(a), we obtain the following:
\begin{corollary}
The most vital nodes problem for Hamilton path, Hamilton cycle, and path cover can be solved in time $\bigO(kn^3)$
\end{corollary}

\subsection{Graph Toughness}
\label{subsection:tough}
In this subsection, we will briefly discuss a related parameter, graph toughness, which was introduced by Chv\'{a}tal~\cite{Chvatal1973} as a measure of graph connectivity. 
A graph $G$ is $t$-tough if $|S| \geq t \cdot c(G - S)$ for any cutset $S$, i.e., $S \subseteq V(G)$ and $c(G - S) > 1$.
In the interdiction problem for graph toughness, we are given shrinking intervals $\I, \I'$, such that $G(\I)$ is $t$-tough, and the interdictor wishes to select $X \subseteq \I$ of size at most $k$, such that $G(\I_X)$ is not $t$-tough, i.e., we have
\[
	\max \{t \cdot c(G_X - S) - |S|  : S \subseteq \I_X \text{ and } c(G_X - S) > 1 \} > 0.
\]

If the set on the left-hand side is empty, then the graph $G_X$ is a complete graph and we define the maximum to be $-\infty$.

Observe that the left-hand side of the above expression is almost identical to the definition of scattering number, except for the factor $t$ in front of $c(G_X - S)$.
Therefore, we can adapt the method to solve the assistance problem for scattering number to solve this problem.
In fact, the only modification to the process $P$ is to replace all instances of $\beta \leftarrow \beta + 1$ by $\beta \leftarrow \beta + t$.
Effectively, each component now contributes $t$ instead of 1 toward the final sum.
Hence, following a completely analogous proof to the proof of Theorem~\ref{thm:ham}(a), we obtain the following:

\begin{corollary}
In the shrink-expand framework, the interdiction problem for graph toughness can be solved in time $\bigO(kn^3)$.
\end{corollary}

\section{Conclusion}
We have introduced a new framework of interdiction and assistance problems for interval graphs based on the shrinking and expanding operations.
In this framework, we have provided algorithms and classified the computational complexity of these problems for many classical parameters, including the shortest path, independence number, and clique number.
Except for the interdiction problem for the clique number which is NP-hard, the others are in P.
We have also formulated a polynomial-time algorithm for the assistance problem for the scattering number, which can be used to solve the interdiction problems for Hamiltonicity.
The interdiction problem for the scattering number, however, remains open.
We also noted that the most vital nodes problems of the above parameters are either special cases of the shrink-expand framework or can be easily solved by relevant problems in the framework. 
In particular, this has resolved an open problem posed by Diner et al.~\cite[(Q2)]{diner2018contractionDeletionBlockers}.

Finally, we note that our polynomial-time results do not generalize to superclasses of interval graphs: For the parameters $\alpha$ and $\omega$, Diner et al. proved that the most vital nodes problem is NP-complete on chordal graphs \cite{diner2018contractionDeletionBlockers}. It is also NP-complete to decide whether a chordal and bipartite graph has a Hamilton cycle \cite{muller1996hamiltonian}, so the most vital nodes problem is NP-complete too in this case. Finally, with respect to shortest path, the most vital nodes problem is NP-complete on general graphs \cite{complexityOfFindingMostVitalNodesShortestPath}, and the complexity on chordal graphs (to the best of our knowledge) is currently unknown. 

\subsubsection*{Acknowledgements.}

Stefan Lendl and Lasse Wulf acknowledge the support of the Austrian Science Fund (FWF):
W1230.

%
%

\bibliographystyle{splncs04}
\bibliography{literature}

\appendix

\section{Answer to question (Q1) by Diner et al.\cite{diner2018contractionDeletionBlockers}}
\label{app:contraction}

Let $G = (V, E)$ be a graph. For a set $X \subseteq E$, we denote by $G/X$ the graph obtained from $G$ by repeatedly contracting an edge from $X$, until no such edge remains. In \cref{sec:independence-interdiction}, we answered the open question (Q2) by Diner et al.\ \cite{diner2018contractionDeletionBlockers}. To do this, we used duality of the parameters of the independence number and the clique cover number. We show that using the same technique and a modification of the proof, we can also positively answer their question (Q1). The question is whether the \emph{contraction blocker} problem on interval graphs for parameter $\pi = \alpha$ can be solved in polynomial time. In this problem, one is given an interval graph $G = (V, E)$ and a threshold $t$, and wishes to find the minimum budget $k$, such that there exists $X \subseteq E$ of size $|X| \leq k$, such that $\alpha(G/X) \leq t$. By standard arguments, it suffices to solve the following problem instead: Given some budget $k \in \N$, compute
\begin{equation}
\min_{X \subseteq E, |X| \leq k}\alpha(G/X). \label{eq:contraction-problem}
\end{equation}  
We first show how to compute expression (\ref{eq:contraction-problem}) for connected interval graphs and then extend the answer to disconnected graphs. Like in \cref{sec:independence-interdiction}, let $\kappa(G)$ denote the clique-cover number of $G$. We also use the following notation in this section: We denote by $G = (V, E)$ some interval graph on $n$ vertices, with interval represention $\I = \fromto{[a_1, b_1]}{[a_n, b_n]}$. The set $\mathcal{B} =  \fromto{b_1}{b_n}$ is the set of all endpoints of intervals. Let $B \subseteq \mathcal{B}$. We call an interval $I \in \I$ \emph{disjoint from $B$}, if $I \cap B = \emptyset$. The following lemma establishes a connection between intervals disjoint to $B$ and problem (\ref{eq:contraction-problem}). Note that if $k = n$,  problem (\ref{eq:contraction-problem}) becomes trivial, so we can assume without loss of generality, that $k < n$.

\begin{lemma}
\label{lemma:contraction}
Let $G$ be a connected interval graph on $n$ vertices. If $k < n$, then
\begin{align*}
&\min_{X \subseteq E, |X| \leq k}\alpha(G/X) \\
= &\min\set{|B| : B \subseteq \mathcal{B}, \text{ at most $k$ intervals from $\I$ are disjoint from $B$}}.
\end{align*} 
\end{lemma}
\begin{proof}
\enquote{$\leq$}: Suppose there is a set $B$ such that at most $k$ intervals are disjoint from $B$. Note that not all intervals are disjoint from $B$, because $k < n$. So there exists at least one interval intersecting $B$ and at most $k$ intervals disjoint from $B$. We describe an edge contraction, which reduces the number of intervals disjoint from $B$: Assume there exists at least one interval disjoint from $B$. Then there exists also a pair of an interval $I$ disjoint from $B$ and an interval $J$ intersecting $B$, such that $I \cap J \neq \emptyset$. Otherwise the set of  intervals disjoint from $B$ and the set of intervals intersecting $B$ induce at least two connected components; a contradiction. Because $I \cap J \neq \emptyset$, there is a corresponding edge in $G$. Contracting this edge yields a new interval graph with interval representation $\I' := \I \setminus \set{I, J} \cup \set{I \cup J}$. Note that the number of intervals disjoint from $B$ in $\I'$ is one less than the number of intervals disjoint from $B$ in $\I$. Therefore, we can repeat this procedure at most $k$ times, to obtain an edge set $X \subseteq E$ of size $|X| \leq k$, and an interval representation $\I''$, such that no interval in $\I''$ is disjoint from $B$ and $G/X = \graph(\I'')$. Using duality of $\alpha$ and the clique cover number $\kappa$, and the fact that no interval in $\I''$ is disjoint from $B$, we have
\[\alpha(G/X) = \kappa(G/X) = \kappa(\graph(\I'')) \leq |B|.  \]

\enquote{$\geq$}: Let $X \subseteq E$ with $|X| \leq k$, and let $\alpha(G/X) = t$. Then also $\kappa(G/X) = t$, so there exists a clique cover with $t$ cliques $C_1, \dots, C_t$. By \cref{prop:cliquecover}, each of the $t$ cliques is equal to (or contained in) the set $C(b) = \set{I \in \I : b \in I}$ for some $b \in \mathcal{B}$. Let $b'_1, \dots, b'_t$ be the corresponding $t$ points, such that $C_i \subseteq C(b'_i)$ for all $i = 1,\dots,t$. Let $B := \fromto{b'_1}{b'_t}$. We can now reverse the edge contractions, going from $G/X$ to $G$, and observe that each decontraction adds at most one interval disjoint to $B$. Hence, $B$ is a set of size $t$, such that at most $k$ intervals in $\I$ are disjoint to $B$. \qed
\end{proof}

\begin{theorem}
\label{thm:contraction-blocker}
Given an interval graph $G$ on $n$ vertices and $k \in \N$, problem (\ref{eq:contraction-problem}) can be computed in $\bigO(kn^2)$ time.
\end{theorem}
\begin{proof}
We distinguish three cases:

\textbf{Case 1:} $G$ is connected, and $k \geq n$. Then \cref{eq:contraction-problem} evaluates to 1.

\textbf{Case 2:} $G$ is connected, and $k < n$. Let $\mathcal{B} = \fromto{b_1}{b_n}$ be the set of endpoints, such that w.l.o.g.\ $b_1 < b_2 < \dots < b_n$. Additionally, let $b_0 := -\infty$ and $b_{n+1} := \infty$. For $i,j \in \fromto{0}{n+1}$ and $i < j$, let 
\[c(i, j) = |\set{k \in [n] : b_i < a_k < b_k < b_j}|. \]
Note that $c(i, j)$ is the number of intervals between $b_i$ and $b_j$ disjoint to $\set{b_i, b_j}$. It is easy to see that the set of all values $c(i,j)$ can be computed in $\bigO(n^2)$ time.

We define for $j \in \fromto{1}{n+1}$ and $k' \in \fromto{0}{k}$:
$F(j, k')$ is the minimum size of a set $|B| \subseteq \fromto{b_1}{b_j}$ subject to $b_j \in B$, and at most $k'$ intervals in $\fromto{[a_1,b_1]}{[a_j, b_j]}$ are disjoint to $B$. Furthermore, define $F(0, k') = 0$ for all $k'$. It is then easy to see that 
\begin{equation}
F(j,k') = \min \{ 1 + F(i, k'-c(i,j)) \colon 0 \leq i < j\text{ and } c(i,j) \leq k' \}.
\end{equation}  
Finally, due to \cref{lemma:contraction}, we have that 
\[\min_{X \subseteq E, |X| \leq k}\alpha(G/X) = F(n+1, k) - 1. \]
(Note that the $-1$ comes from the fact that $F(n+1, k)$ counts the point $b_{n+1}$.)

\textbf{Case 3:} $G$ is disconnected. Assume $G$ consists of connetced components $C_1, \dots, C_{\ell}$, where component $C_i$ has $n_i$ vertices. The interdictor needs to distribute the budget $k = k_1 + \dots + k_{\ell}$ to the $\ell$ components. By the above argument, we can precompute in time $\bigO(kn_i^2)$ a table of all the values $\min_{|X| \leq k_i} \alpha(C_i/X)$ for $k_i \in \fromto{1}{k}$. The interdictor now needs to solve 
\[ \min_{k = k_1 + \dots + k_{\ell}} \sum_{i=1}^{\ell} \min_{|X| \leq k_i} \alpha(C_i/X). \]
It is an easy exercise to prove that this formula can be evaluated in time $\bigO(k\ell)$ using dynamic programming. In total, we have a running time of $\bigO(\sum_i kn_i^2 + k\ell) = \bigO(kn^2)$. \qed
\end{proof}

\begin{corollary}
The contraction blocker problem on interval graphs for parameter $\pi = \alpha$ can be solved in time $\bigO(n^3)$.
\end{corollary}
\begin{proof}
Running the algorithm in \cref{thm:contraction-blocker} for $k = n$, we obtain a dynamic programming table containing for each $k' \in \fromto{1}{n}$ the maximum possible effect the interdictor can achieve with a budget of $k'$. Using binary search, we can find the minimum budget $k'$ necessary to reduce $\alpha$ down to the threshold $t$.
\end{proof}
\end{document}